\documentclass[a4paper,12pt]{article}
\usepackage{amsmath,amssymb,amsthm}
\usepackage{fullpage}
\usepackage{color}
\usepackage{tikz}

\newtheorem{construction}{Construction}
\newtheorem{theorem}{Theorem}[section]
\newtheorem{corollary}[theorem]{Corollary}

\newtheorem{problem}{Question}
\theoremstyle{definition}
\newtheorem{definition}{Definition}[section]
\newtheorem{example}{Example}[section]

\newcommand{\other}{\ensuremath{\mathrm{other}}}
\newcommand{\Exp}{\mathrm{E}}
\newcommand{\bv}{\mathbf{v}}
\newcommand{\bx}{\mathbf{x}}
\newcommand{\bZ}{\mathbb{Z}}

\begin{document}
\title{PIR schemes with small download complexity\\
and low storage requirements\thanks{Parts of this paper were presented at the International Symposium on Information Theory, Aachen, Germany, June 2017}}
\author{Simon R. Blackburn\thanks{Department of Mathematics, Royal Holloway University of London,
Egham, Surrey TW20 0EX, United Kingdom, e-mail: {\tt s.blackburn@rhul.ac.uk}.}\and Tuvi Etzion\thanks{Department of Computer Science, Technion,
Haifa 3200003, Israel, e-mail: {\tt etzion@cs.technion.ac.il}. Part of the
research was performed while the second author visited Royal Holloway University of London under EPSRC
Grant EP/N022114/1. Part of the research was supported by the NSF-BSF grant 2016692.}\and Maura B. Paterson\thanks{Department of Economics, Mathematics and Statistics, Birkbeck, University of London, Malet Street, London WC1E 7HX, United Kingdom, e-mail: \texttt{m.paterson@bbk.ac.uk}.}}
\maketitle

\vspace{-1cm}

\begin{abstract}
In the classical model for (information theoretically secure) Private Information Retrieval (PIR) due to Chor, Goldreich, Kushilevitz and Sudan, a user wishes to retrieve one bit of a database that is stored on a set of $n$ servers, in such a way that no individual server gains information about which bit the user is interested in. The aim is to design schemes that minimise the total communication between the user and the servers. More recently, there have been moves to consider more realistic models where the total storage of the set of servers, or the per server storage, should be minimised (possibly using techniques from distributed storage), and where the database is divided into $R$-bit records with $R>1$, and the user wishes to retrieve one record rather than one bit. When $R$ is large, downloads from the servers to the user dominate the communication complexity and so the aim is to minimise the total number of downloaded bits. Work of Shah, Rashmi and Ramchandran shows that at least $R+1$ bits must be downloaded from servers in the worst case, and provides PIR schemes meeting this bound. Sun and Jafar have considered the download cost of a scheme, defined as the ratio of the message length $R$ and the total number of bits downloaded. They determine the best asymptotic download cost of a PIR scheme (as $R\rightarrow\infty$) when a database of $k$ messages is stored by $n$ servers.

This paper provides various bounds on the download complexity of a PIR scheme, generalising those of Shah et al.\ to the case when the number $n$ of servers is bounded, and providing links with classical techniques due to Chor et al. The paper also provides a range of constructions for PIR schemes that are either simpler or perform better than previously known schemes. These constructions include explicit schemes that achieve the best asymptotic download complexity of Sun and Jafar with significantly lower upload complexity, and general techniques for constructing a scheme with good worst case download complexity from a scheme with good download complexity on average.
\end{abstract}

\section{Introduction}

\subsection{The PIR Model}

In the classical model for private information retrieval (PIR) due to Chor, Goldreich, Kushilevitz and Sudan~\cite{CGKS98},
a database $\mathbf{X}$ is replicated across $n$ servers $S_1,S_2,\dotsc,S_n$.
A user wishes to retrieve one bit of the database, so sends a query to each
server and downloads their reply. The user should be able to deduce the bit
from the servers' replies. Moreover, no single server should gain any information
on which bit the user wishes to retrieve (without collusion). The resulting protocol
is known as an (information-theoretic) \emph{PIR scheme}; there are also computational
variants of the security model~\cite{KuOs97}. The goal of PIR is to minimise
the total communication between the user and the servers.

In practice, the assumption that the user only wishes to retrieve one bit of the database, and the assumption that there is no shortage of server storage seem unrealistic. Because of this, many recent papers assume that the database $\mathbf{X}$ consists
of $k$ {\em records}, each of which is $R$ bits in length, so that the number
of possible databases is $2^{kR}$.  We denote the value of Record $i$
by $X_i$, and we write $X_{ij}$ for the $j$th bit of $X_i$. The aim of the protocol
is for the user to retrieve the whole of $X_i$, rather than a single bit. We also, following Shah, Rashmi and Ramchandran~\cite{SRR14}, drop the assumption that the whole database is replicated across the $n$ servers $S_1,S_2,\ldots,S_n$ and so, for example,
there is the possibility of using techniques from coding theory
in general and from distributed storage codes in particular to reduce the total storage of the scheme.
No restrictions are made on the particular encoding used to distribute the database across the servers other than to assume it is deterministic, {\it i.e.}~that there is a unique way to encode each database. This important generalisation of the model has led to very interesting recent work which we discuss in Subsection~\ref{subsec:context} below.
Our work is a follow-up to~\cite{SRR14} with modifications, improvements, complementary results, simplifications, constructions,
and additional aspects which were not considered in their paper.

More combinatorially, we define a private information retrieval scheme as follows.
\begin{definition}[PIR scheme]
Suppose a database $\mathbf{X}$ is distributed across $n$ servers $S_1,S_2,\dotsc,S_n$.
A user who wishes to learn the value $X_\ell$ of Record $\ell$ submits a {\em query} $(q_1,q_2,\dotsc,q_n)$.
For each $r\in \{1,2,\dotsc, n\},$ server $S_r$ receives $q_r$ and responds with a
value $c_r$ that depends on $q_r$ and on the information stored by $S_r$. The user
receives the {\em response} $(c_1,c_2,\dotsc,c_n)$.  This system is a
{\em private information retrieval (PIR) scheme} if the following two properties are satisfied:
\begin{description}
\item \textbf{(Privacy)} For $r=1,2,\dotsc,n$ the value $q_r$ received by
server $S_r$ reveals no information about which record is being sought.
\item \textbf{(Correctness)} Given a response $(c_1,c_2,\dotsc,c_n)$ to a
query $(q_1,q_2,\dotsc,q_n)$ for Record $\ell$, the user is unambiguously able to recover the value $X_\ell$ of this record.
\end{description}

\end{definition}
Note that while the query is drawn randomly according a pre-specified distribution
on a set of potential queries, the response is assumed to be deterministic.
\begin{example}
\label{ex:1server}
In the case of a single server, a trivial method for achieving PIR is for the user to download the entire $kR$-bit database.
\end{example}
Chor, Goldreich, Kushilevitz and Sudan showed that in the case of single-bit
records ($R=1$),  if there is a single server then PIR is only possible if
the total communication is at least $k$ bits ({\it i.e.}~the size of the
entire database) \cite{CGKS98}, and so the solution above is best possible.
We are interested in finding solutions such as the scheme below,
which transmit significantly fewer than $kR$ bits.

\begin{example}\cite{CGKS98}
\label{ex:chor}
Suppose there are two servers, each storing the entire database. Suppose $R=1$.
\begin{itemize}
\item A user who requires Record $\ell$ chooses a $k$-bit string $(\alpha_1,\alpha_2,\dotsc, \alpha_k)$ uniformly at random.
\item Server 1 is requested to return the value $c_1=\bigoplus_{i=1}^k \alpha_i X_i$, and Server 2 is requested to return $ c_2=\left(\bigoplus_{i=1}^k \beta_i X_i\right)$, where
\[
\beta_i=\begin{cases}
\alpha_i\oplus 1&\text{when }i=\ell,\\
\alpha_i&\text{otherwise.}
\end{cases}
\]
\item The user computes $c_1\oplus c_2$ to recover the value $X_\ell$ of Record $\ell$.
\end{itemize}
\end{example}

The strings $(\alpha_1,\alpha_2,\dotsc, \alpha_k)$ and $(\beta_1,\beta_2,\dotsc, \beta_{k})$ are both uniformly distributed, and are independent of the choice of $\ell$, hence neither server receives any information as to which record is being recovered by the user.

We note that the scheme above works unchanged when the records are $R$-bit strings rather than single bits.
The download complexity of the scheme, in other words the total number of bits
downloaded from the servers, is $2R$. The upload complexity is $2k$, since each server
receives a $k$-bit string from the user. Thus the total communication of the
scheme is $2R+2k$ bits, which is significantly less than $kR$ bits for most parameters.

Note that the upload complexity of this scheme does not depend on $R$, and so is an insignificant
proportion of the total communication when $R$ is large. This is a general phenomenon:
Chan, Ho and Yamamoto~\cite[Remark 2]{CHY15} observe the following. Let $m>1$ be an integer.
Suppose we have an $n$-server PIR scheme for a database of $k$ records,
each $R$ bits long. Suppose the scheme requires $u$ upload bits and $d$ download bits.
Then we can construct an $n$-server PIR scheme for a database of $k$ records,
each $mR$ bits long, which requires $md$ bits of download but still needs just $u$ bits to be uploaded.
Note that when $m$ is large (so records are long) the communication complexity
of the new scheme is dominated by the download complexity of the given scheme.

Because of the observation of Chan et al., it is vital to find PIR schemes with low download complexity.
We formalise the download complexity as follows.

\begin{definition}
A PIR scheme \emph{uses binary channels} if the response $c_j$ sent by server $S_j$ is a binary string
of length $d_j$, where $d_j$ depends only on the query $q_j$ it receives. The \emph{download complexity}
is the maximum of the sum $\sum_{j=1}^n d_j$ over all possible queries $(q_1,q_2,\dotsc,q_n)$.
\end{definition}
So the download complexity is the number of bits downloaded in the worst case. We emphasise that
the length $d_j$ in the definition above does not depend on the database $\mathbf{X}$,
but could depend on the query $q_j$ received by server $S_j$. We note that we allow
for the possibility that $d_j=0$, so the server does not reply to the query.
Finally, we note that if we know that there are more than $2^x$ distinct
possibilities for $c_j$ as the database varies, we may deduce that $d_j\geq x+1$.
Although it is possible to use non-binary channels for PIR as was done recently
(see Subsection~\ref{subsec:context}), we restrict our exposition to PIR schemes
using binary channels in this paper. Most schemes in the literature before the work of~\cite{SRR14} implicitly use this model,
as they use fields of characteristic $2$ and transmit strings of bits.
This restricted model is implicit also in Shah et al.~\cite{SRR14}
when a lower bound on the download complexity is proved. This model is only required
for those results in Section~\ref{sec:lower_bounds} paper that are used to generalise their bound.
Although most results in the other sections can be generalised for the non-binary case we prefer not to do
so for simplicity.

We should comment that, despite the observation of Chan et al., we should not ignore upload complexity completely, as there are scenarios (for example, when $R$ is not so large)
when it might be dominant. Moreover, we cannot compare the difficulty of uploading $2k$ bits with downloading $2R$ bits just by comparing $k$ and $R$, since we know (at least currently) that in practice it takes
much more time to upload a bit than to download one. Of course, we don't know how the speed of
downloading and uploading will change over time. But the obvious consequence of the current situation and future developments is to consider both upload and download complexities separately, and not to ignore one of them completely. This is something that will be done in this paper, although the download complexity will be the main target for optimisation since we generally assume that the size of the database is (considerably) larger than the one bit of the classical PIR model.

We continue to another important measure that has motivated many papers in the last three years, after being introduced by Shah et al.~\cite{SRR14}:
\begin{definition}
Suppose server $S_r$ stores $s_r$ bits of information about the database $\mathbf{X}$.
\begin{itemize}
\item The \emph{per-server storage} of the scheme is $\max\{ s_r\mid r=1,2,\ldots ,n\}$.
\item The \emph{total storage} of the scheme is $\sum_{r=1}^ns_r$.
\item The \emph{storage overhead} of the scheme is the ratio between the total storage and
the total size of the records in the database, i.e. $kR$.
\end{itemize}
\end{definition}

The classical model of PIR ignored storage issues: it was assumed that there is
enough storage to allow the replication of the database at each server. But, with the quantity of information stored today in data centres, storage is an issue today and might
be an important barrier in the future. Thus, it is important to reduce the storage overhead as much
as possible, while keeping reliability, fast access, fast upload and fast download at reasonable levels.
This is the perspective of the current paper, which is concerned with schemes whose download complexity is as small as possible whilst keeping the total storage at reasonable levels.

Finally, it should be noted that although most of the work in this area is theoretical,
there have been notable recent advances in
bridging the gap between theory and practice, e.g.~\cite{GCMSAW16,WYGVZ17}
as we highlight in Subsection~\ref{subsec:context}.

\subsection{Our contributions}

In Section~\ref{sec:lower_bounds}, we provide combinatorial results on the structure of a PIR scheme with small download complexity:
\begin{itemize}
\item We generalise (Theorem~\ref{thm:download}) a key theorem in the foundational paper of Chor et al.~\cite{CGKS98}, and use this result to generalise
the lower bound of $R+1$ on download complexity in~\cite{SRR14}. The results imply (Theorem~\ref{thm:tdownload}) that an $n$-server PIR scheme must have download complexity at least $\frac{n}{n-1}R$ when $k>\lceil R/(n-1)\rceil$. (This last result can also be obtained as a corollary of a recent bound due to Sun and Jafar~\cite{SuJa16b}.) These results provide a bridge between classical PIR and the new models that are assuming the retrieval of long records. Moreover (as often happens with a combinatorial approach), some extra structural information on schemes is provided: see Theorem~\ref{thm:binarydownload}.
\item We provide (Corollary~\ref{cor:onebit}, Theorem~\ref{thm:almostall}) information on the structure of a PIR scheme with minimal download complexity $R+1$. In particular, Theorem~\ref{thm:almostall} provides a rigorous statement of~\cite[Theorem~1]{SRR14}.
\end{itemize}

In Section~\ref{sec:constructions}, we provide various constructions for PIR schemes with low download complexity:
\begin{itemize}
\item In Subsection~\ref{subsec:rplusone}, we provide two simple $(R+1)$-server PIR schemes with download complexity $R+1$. Both schemes have total storage which is quadratic in $R$. The first scheme is a natural generalisation of the scheme of Chor et al.\ given above. The second scheme is a close variant of the quadratic total storage PIR scheme in~\cite{SRR14}, which avoids having to design slightly different schemes depending on the parity of $R$. This second scheme is to be preferred due to its lower upload complexity. (Another, more complex, PIR scheme with download complexity $R+1$ is considered in detail in~\cite{SRR14}. This scheme has small per-server storage,
but requires an exponential (in $R$) number of servers, and so has exponential total storage.)
\item In Subsection~\ref{subsec:nserver}, we describe an $n$-server PIR scheme with download complexity $\frac{n}{n-1}R$. The total storage of the scheme is linear in $R$. This shows that for any $\epsilon>0$ there exists a PIR scheme with linear total storage and download complexity at most $(1+\epsilon)R$. (Schemes with linear total storage, but with download complexity between $2R$ and $4R$, are given in~\cite{SRR14}.)
\item We describe (Subsection~\ref{subsec:perserver}) schemes that provide trade-offs between increasing the number of servers and reducing the per-server storage of the scheme in Subsection~\ref{subsec:nserver}.
\item In Subsection~\ref{subsec:sunjafar}, we provide explicit schemes that achieve optimal asymptotic download cost. The performance of these schemes is equal to the inductively defined schemes in Sun and Jafar~\cite{SuJa16b}, but the description of these schemes is more concise, and the proof that they are indeed PIR schemes is much more straightforward.
\item Finally, in Subsection~\ref{subsec:averaging}, we explain an averaging technique that allows a PIR scheme with good average download complexity to be transformed into a scheme with good download complexity in the worst case.
\end{itemize}

\subsection{Context}
\label{subsec:context}

We end this introduction with a discussion of some of the related literature.
(Many of these papers appeared after the conference version of our paper~\cite{BEM17}
was posted. We omit results submitted after our submission of this paper.)

Private information retrieval was introduced in~\cite{CGKS98}, and has been
an active area ever since. See, for example, Yekhanin~\cite{Y10} for a fairly recent survey.

The papers by Shah et al.~\cite{SRR14} and (independently) by Augot, Levy-Dit-Vahel, and Shikfa~\cite{ALS14}
are the first to consider PIR models where the information stored
in the servers could be coded using techniques from distributed storage.
Whereas~\cite{SRR14} is mainly concerned with download complexity, and also with total storage
(with per-server storage, and query size also relevant parameters),
the authors of~\cite{ALS14} emphasise measures of robustness against
malicious servers, namely decoder locality and PIR locality.

More recently, the literature has addressed several parallel and related
issues, which can be categorised as follows:

\begin{enumerate}
\item Papers dealing with the download complexity, rate, and capacity of PIR schemes.

\item Research which attempts to reduce the storage overhead of PIR schemes.

\item Papers which present coding techniques, based on various error-correcting codes,
e.g. MDS codes, to store the database in a distributed fashion.

\item Papers which consider PIR schemes in the presence of unreliable servers. Servers might be colluding (so they have access to more than one query $q_r$), they might fail (and so do not reply with a value $c_r$), they might be adversarial (replying with incorrect values $c_r$), they might be unsynchronised (storing slightly different copies of the database) and so on.
\item Research which aims to build PIR schemes into previously known architectures for distributed storage.
\item Papers dealing with other PIR models, for example allowing broadcasting of some information, or allowing the user to possess side information such as the value of some records.
\end{enumerate}
Clearly, these issues are related, and a given paper might address aspects of more than one of these topics.

In early papers, Fanti and Ramchandran~\cite{FaRa14,FaRa15} considered unsynchronized databases; the results
are the same as for synchronized PIR at the expense of probabilistic success for
information retrieval, and the use of two rounds of communication. We are not aware of recent work in this model, but we mention in this context the work of Tajeddine and El Rouayheb~\cite{TaElR17}
which considers PIR schemes in the presence of some servers which do not respond to a query.

In a sequence of papers, Sun and Jafar~\cite{SuJa16b,SuJa16e,SuJa16c,SuJa16f,SuJa16d} consider
the capacity of the channels related to PIR codes in various models. (The \emph{rate} of a PIR scheme is the ratio of $R$ and the download complexity, and the \emph{capacity} is the supremum of achievable
rates.) In the model for PIR we consider, they use information theoretic techniques to show~\cite{SuJa16b}
that an $n$-server PIR scheme on a $k$ message database has rate at most
\[
\left(1-\frac{1}{n}\right)\left(1-\left(\frac{1}{n}\right)^k\right)^{-1}.
\]
Their model is restricted to the special case of replication.
They also provide a scheme that attains this rate. The messages in their scheme
are extremely long for most values of $n$ and $k$: the message length must
be a multiple of $n^k$. Because of this, the scheme can be thought of as being
tailored for the situation when $R\rightarrow\infty$. Their results show that
when $R\rightarrow\infty$ with $n$ and $k$ fixed, there are schemes whose download
complexity (and so whose communication complexity) has a leading term of the form
\[
\frac{n}{n-1}\left(1-\left(\frac{1}{n}\right)^k\right)R,
\]
and that this term is best possible. (We give an explicit scheme with the same download complexity in Subsection~\ref{subsec:sunjafar}.)

The results in~\cite{SuJa16b} have been generalised to the case when some of the servers
collude. Sun and Jafar~\cite{SuJa16c} find the capacity of the channel in this more general
case. (The results in~\cite{SuJa16b} can be thought of as the special case where each server can collude only with itself.)
The capacity for the symmetric PIR model, where the user
who retrieves a message will get no information about the other messages in the database,
is determined in~\cite{SuJa16d}. The optimal download complexity in the situation
when the messages in the database might be of an arbitrary length (subject to a certain divisibility condition
with all messages having the same length) is considered
in~\cite{SuJa16e}. The most recent in this sequence of papers considers an interactive model, where
a user can have several rounds of queries, the queries in a given round are allowed to depend on answers from previous rounds. Moreover, colluding servers are considered in this model.  It is proved~\cite{SuJa16f} that for this case there is no change in the capacity, but that the storage overhead can sometimes be improved.

Banawan and Ulukus~\cite{BaUl16} also generalise the results of Sun and Jafar~\cite{SuJa16b},
finding the exact capacity of the PIR scheme when the database is encoded with
a linear code. Another generalisation due to Banawan and Ulukus~\cite{BaUl17a,BaUl17} is
to the scenario that the user is allowed to request a few records in one round of queries.
They provide capacity computations and schemes for this scenario.
A similar case was also discussed in~\cite{ZhXu17}.
Finally, Banawan and Ulukus~\cite{BaUl17b} consider the capacity of PIR schemes
in the scenario where servers might not be synchronised,
there might be adversarial errors, and some servers might collude. They
compute the capacity when some or all of these events might occur.
Wang and Skoglund~\cite{WaSk16} consider the capacity of a symmetric
PIR scheme when the database is stored in a distributed fashion using an MDS code.

Chan, Ho, and Yamamoto~\cite{CHY14,CHY15} consider the trade-off between the total storage
and the download complexity when the size of a record is large;
the trade-off depends on the number of records in the system.
They also consider the case where the database is encoded with an MDS code.

Fazeli, Vardy, and Yaakobi~\cite{FVY15,FVY15a} give a method to reduce
the storage overhead based on any known PIR scheme which uses replication.
Their method reduces the storage overhead considerably, without affecting
the order of the download complexity or upload complexity of the overall scheme, by simulating the original scheme on a larger number of servers.
Their key concept is an object they call a $\kappa$-PIR code
(more generally a $\kappa$-PIR array code), where
$\kappa$ is the number of servers used in the originally known PIR scheme, which controls how a database can be divided into parts and encoded within servers to allow a trade-off between the number of servers
and the storage overhead. In particular, for all $\epsilon>0$, they show that there exist good schemes
(in terms of communication requirements) where the amount of information stored in a server
is bounded but the total storage is at most $(1+\epsilon)$ times the database size.
Rao and Vardy~\cite{RaVa16} study PIR codes further, establishing the asymptotic
behavior of $\kappa$-PIR codes. Vajha, Ramkumar, and Kumar~\cite{VRK17,VRK17a}
find the redundancy of such codes for $\kappa =3,~4$ by using Reed-Muller codes. Lin and Rosnes~\cite{LiRo17}
show how to shorten and lengthen PIR codes, and find the redundancy of such codes for $\kappa =5,~6$.
Blackburn and Etzion~\cite{BE16,BlEt17} consider the optimal ratios between $\kappa$-PIR array codes
and the actual number of servers used in the system. Zhang, Wang, Wei, and Ge~\cite{ZWWG16} consider
these ratios further, and improve some of the results from~\cite{BE16,BlEt17}.
We remark that though it is possible to reduce the storage overhead using the techniques of PIR array codes,
it seems impossible to reduce the download complexity of the resulting schemes below $(3/2)R$
(and most codes give download complexity close to $2R$) because of restrictions on the PIR rate
of such codes. It is interesting to note that Augot, Levy-Dit-Vahel, and Shikfa~\cite{ALS14}
constructed PIR schemes by partitioning the database into smaller parts, as done
later in~\cite{FVY15,FVY15a}, to
reduce the storage overhead. But they applied this technique only to a certain family of multiplicity codes, and the parts of the partition were not encoded as in~\cite{FVY15,FVY15a}.

Fazeli, Vardy, and Yaakobi~\cite{FVY15a} remark that the concept of a $\kappa$-PIR code
is closely related to codes with locality and availability. Such codes were studied first
by Rawat, Papailiopoulos, Dimakis, and Vishwanath~\cite{RPDV14b,RPDV16} and later also
by others, for example~\cite{FFGW17,HYUP15}. A new subspace approach for such codes was given recently in~\cite{SES17a,SES17}. Another family of related codes with similar properties are batch codes, which were first defined
by Ishai, Kushilevitz, Ostrovsky, and Sahai~\cite{IKOS04} and were recently studied
by many others, for example~\cite{AsYa17,AsYa18,RSDG16}.
It is important to note that all these codes are very important in the theory
of distributed storage codes. This connection between the concepts of locality and PIR codes
are explored in~\cite{FFGW17}.

Error-correcting codes, and in particular maximum distance separable (MDS) codes, have been considered by many authors in various
PIR models. It is natural to consider MDS codes, as they are very often used in various types of
distributed storage codes (especially for locally repairable codes~\cite{GHSY12}
and regenerating codes~\cite{DGWWR10,DRWS11}), and we expect that
the servers in our PIR scheme will be part of a distributed storage system. We will now mention various examples.

Colluding or malicious servers in PIR have been much studied over the last two years.
Tajeddine and El Rouayheb~\cite{TaElR16a} consider PIR schemes
where the information is stored using MDS codes.
Their PIR scheme based on the coded MDS
achieves a retrieval rate $1-R$, where $R$ is the code rate of the storage system.
They attain the bounds for linear schemes in~\cite{CHY14,CHY15},
in the situation when one or two `spies' (colluding and/or malicious servers) are present.
In the case of one spy (no collusion) a generalisation to any linear code
with rate greater than half was given in~\cite{KRA16}.
Freij-Hollanti, Gnilke, Hollanti, and Karpuk~\cite{FGHK16} give a PIR scheme coded with
an MDS code which can be adjusted (by varying the rate of the MDS code) to combat against
larger numbers of colluding servers. This scheme also attains the asymptotic bound
on the related capacity of such a PIR scheme in the extreme cases, where
there are no colluding servers or when the database is replicated, i.e. no coding is applied. This idea is generalised
in~\cite{TGKFHR17}. The results
in the latter paper are analysed (and one conjecture disproved)
by Sun and Jafar~\cite{SuJa17a,SuJa17}.
Another scheme based on MDS codes which can combat large number of colluding servers
is given by Zhang and Ge~\cite{ZhGe17a}. A generalisation to the case where the user wants to
retrieve several files is given by the same authors in~\cite{ZhGe17b}.
Wang and Skoglund~\cite{WaSk17a} consider a symmetric PIR scheme using an MDS code,
in which the user can retrieve the information about the file he wants,
but can gain no information about the other files. This scheme attains
the bound on the capacity which they derive earlier in~\cite{WaSk16}.
They have extended their work to accommodate colluding servers in~\cite{WaSk17b}.

PIR can be combined with other applications in storage and communication
in many ways. One example is a related broadcasting scheme in~\cite{KSCF17}.
Another example is cache-aided PIR, considered by Tandon~\cite{Tan17}. In this setup the user is equipped with a local cache which is formed from an arbitrary
function on the whole set of messages, and this local cache is known to the servers.
The situation when this cache is not known to the servers is considered
by Kadhe, Garcia, Heidarzadeh, El Rouayheb, and Sprintson~\cite{KGHRS17}. Since the user has side information in these models,
the problem is closely related to index coding~\cite{BBJK11} a topic which is also of great interest.

While most of the work in this area is theoretical, there have been notable recent advances in
bridging the gap between theory and practice. For example, the recent paper~\cite{GCMSAW16} reports on the design
and implementation of a scalable and private media delivery system --- called \emph{Popcorn} --- that
explicitly targets Netflix-like content distribution. Another practical system for private queries on
public datasets --- called \emph{Splinter} --- is currently in development~\cite{WYGVZ17}.
This system has been reported to achieve latencies below 1.20 seconds for realistic
workloads including a Yelp clone, flight search, and map routing.

\section{Optimal download complexity}
\label{sec:lower_bounds}

In this section, we give structural results for PIR schemes with optimal download complexity, given that the database consists of $k$ records of length $R$. For some of the results, we also assume that the PIR scheme involves $n$ servers, where $n$ is fixed.

In Subsection~\ref{subsec:lower_bounds} we generalise a classical result in Private Information Retrieval due to Chor et al. We use this result to provide an alternative proof of the theorem of
Shah, Rashmi and Ramchandran~\cite{SRR14} that a PIR scheme must have download complexity at least $R+1$ when $k\geq 2$, and to prove a lower bound of $\frac{n}{n-1}R$ for the download complexity of an $n$-server PIR scheme whenever $k$ is sufficiently large. In Section~\ref{subsec:R_plus_one} we present more precise structural results when the download complexity of a PIR scheme attains the optimal value of $R+1$ bits.

\begin{definition}
We say that a response $(c_1,c_2,\dotsc,c_n)$ is {\em possible} for a query $(q_1,q_2,\dotsc,q_n)$ if there exists a database $\mathbf{X}$  for which $(c_1,c_2,\dotsc,c_n)$ is returned as the response to the query $(q_1,q_2,\dotsc,q_n)$ when $\mathbf{X}$ is stored by the servers.
\end{definition}

\subsection{Lower bounds on the download complexity}
\label{subsec:lower_bounds}

We aim to generalise the following theorem, which was proved by Chor et al.\ in the very first paper on PIR~\cite[Theorem~5.1]{CGKS98}:

\begin{theorem}
\label{thm:chor}
A  PIR scheme that uses a single server for a database with $k$ records of size one bit is not possible unless the number of possible responses from the server to any given query is at least $2^k$.
\end{theorem}

Our generalisation shows a server must reply with at least $k(R-d)$ bits of download, if no more than a total of $d$ bits (where $0\leq d \leq R$) are downloaded from the other servers. We state our generalisation as follows. Without loss of generality we will focus on server $S_1$, so for ease of notation we will denote the tuple $(q_1,q_2,\dotsc,q_n)$ by $(q_1,q_\other)$, and $(c_1,c_2,\dotsc, c_n)$ by $(c_1,c_\other)$.

\begin{theorem}\label{thm:download} Suppose $0\leq d \leq R$.  Let $q_1$ be fixed. Suppose we have a PIR scheme with the property that for any query of the form $(q_1,q_\other)$, we have
\begin{equation*}
|\{c_\other\mid\text{$\exists c_1$ such that $(c_1,c_\other)$ is possible for $(q_1,q_\other)$}\}|\leq 2^{d}.
\end{equation*}
Then for any query $(q_1,q'_\other)$ we have
\begin{equation*}
|\{c_1\mid\text{$\exists c_\other$ such that $(c_1,c_\other)$ is possible for $(q_1,q'_\other)$}\}|\geq 2^{k(R-d)}.
\end{equation*}
\end{theorem}
We remark that Theorem~\ref{thm:chor} is the case $d=0$ and $R=1$ of Theorem~\ref{thm:download}.
\begin{proof}
Let $q_1$ be fixed, and suppose we have a PIR scheme with the property that for any query $(q_1,q_\other)$
\begin{equation}
\label{eqn:reply}
|\{c_\other\mid\text{$\exists c_1$ such that $(c_1,c_\other)$ is possible for $(q_1,q_\other)$}\}|\leq 2^{d}.
\end{equation}
Assume, for a contradiction, that there exists a query $(q_1,q_\other^\ast)$ for which
\begin{equation*}
|\{c_1\mid\text{$\exists c_\other$ such that $(c_1,c_\other)$ is possible for $(q_1,q_\other^\ast)$}\}|<2^{k(R-d)}.
\end{equation*}
Suppose this query is for Record~$i$.

Let $c_1^\ast$ be a most common reply of $S_1$ to $(q_1,q_\other^\ast)$ as the database varies over all possibilities. So we choose $c_1^\ast$  to maximise $|T|$, where $T$ is the set of databases where $S_1$ replies with $c_1^*$ to the query $(q_1,q_\other^\ast)$. If server $S_1$ receives the query $q_1$, it will thus return $c_1^\ast$ whenever a database in $T$ is being stored. There are $2^{kR}$ databases, and less than $2^{k(R-d)}$ possibilities for the reply $c_1$ of $S_1$ to the query $(q_1,q_\other^\ast)$. So by the pigeonhole principle, $|T|>2^{kR}/2^{k(R-d)}=2^{kd}$.

Since the databases consist of $k$ records, the fact that $|T|>2^{kd}$ implies the existence of a record, say Record $\ell$, for which the number of distinct values $X_\ell$ that appear among the databases in $T$ is greater than $2^d$.  Thus we can choose a subset of $2^d+1$ databases $W\subseteq T$ such that the values $X_\ell$ of Record~$\ell$ in the databases in $W$ are all distinct.

The requirement for privacy against server $S_1$ implies that there exists a query for Record $\ell$ of the form $(q_1,q_\other^\ell)$, since otherwise $S_1$ could distinguish between queries for Record $i$ and Record $\ell$.

Suppose the query $(q_1,q_\other^\ell)$ for Record $\ell$ is made, and suppose that the database lies in $W$. Server $S_1$ receives $q_1$, and so (since $W\subseteq T$) replies with $c_1^*$. But there are at most $2^d$ possible replies $c_\other^\ell$ from the remaining servers by~\eqref{eqn:reply}, and so there are at most $2^d$ responses $(c_1^\ast,c_\other^\ell)$ to the query $(q_1,q_\other^\ell)$. Since $|W|=2^d+1$, there are two databases $\mathbf{X},\mathbf{Y}\in W$ such that the servers respond identically. But this is our required contradiction, since $\mathbf{X}$ and $\mathbf{Y}$ have distinct values for Record~$\ell$ and the query was for this record.
\end{proof}

The following theorem is a key consequence of Theorem~\ref{thm:download}.

\begin{theorem}\label{thm:binarydownload}
Let $x$ be non-negative, and suppose we have a PIR scheme that has download complexity at most $R+x$. If the database contains $k$ records, where $k\geq x+2$, then the number of bits downloaded from any server is at most~$x$.
\end{theorem}
\begin{proof}
Without loss of generality, consider the server $S_1$. Suppose for a contradiction that there exists a query $q_1$ so that at least $x+1$ bits are downloaded from $S_1$ (and so at most $(R+x)-(x+1)=R-1$ bits are downloaded from the other servers). Suppose that a total of $d$ bits are downloaded from the other servers in the worst case when $S_1$ receives~$q_1$. So $d\leq R-1$. Theorem~\ref{thm:download} implies that at least $k(R-d)$ bits are downloaded from $S_1$, and so at least $k(R-d)+d$ bits are downloaded from the servers in the worst case. But $d\leq R-1$ and $k\geq x+2$, so
\[
k(R-d)+d=kR-(k-1)d
\geq kR-(k-1)(R-1)=R+k-1\geq R+(x+2)-1
=R+x+1,
\]
which is impossible as the scheme has total download complexity $R+x$. This contradiction establishes the theorem.
\end{proof}

We are now in a position to provide a new short proof of the following corollary. The corollary is due to Shah et al.~\cite{SRR14}.

\begin{corollary}
\label{cor:rplusone}
Let the database contain $k$ records with $k\geq 2$. Any PIR scheme requires a total download of at least $R+1$ bits.
\end{corollary}
\begin{proof}
Suppose we have a scheme with total download of $R$ or fewer bits. Theorem~\ref{thm:binarydownload} with $x=0$ implies that $0$ bits are downloaded from each server, and so the user receives no information about the desired record. Hence such a scheme cannot exist.
\end{proof}

The following theorem (which can also be derived from the results in~\cite{SuJa16b}), improves the bound of Corollary~\ref{cor:rplusone} when $n< R+1$ and $k$ is sufficiently large.

\begin{theorem}\label{thm:tdownload}
Suppose a PIR scheme involves $n$ servers, where $n\geq 2$. Suppose the database contains $k$ records, where $k\geq \lceil \frac{1}{n-1}R\rceil+1$. Then the download complexity of the scheme is at least $\frac{n}{n-1}R$ bits.
\end{theorem}
\begin{proof}
Assume for a contradiction that the scheme has download complexity $R+x$, where $x$ is an integer such that $x<\frac{1}{n-1}R$. Since $x\leq \lceil \frac{1}{n-1}R\rceil -1$, we see that $k\geq x+2$ and so Theorem~\ref{thm:binarydownload} implies that the number of bits downloaded by any server is at most $x$. Since we have $n$ servers, the total number of bits of download is always at most $xn$. Since our scheme has download complexity $R+x$, there is a query where a total of $R+x$ bits are downloaded from servers. Hence we must have that $nx\geq R+x$, which implies that $x\geq \frac{1}{n-1}R$. This contradiction establishes the result.
\end{proof}

\subsection{Download complexity $R+1$}
\label{subsec:R_plus_one}

The final two results of this section concentrate on the extreme case when the download complexity is exactly $R+1$. Recall that the download complexity is a worst case measure: every query results in at most $R+1$ bits being downloaded, and there exists a query where $R+1$ bits are downloaded.

\begin{corollary}
\label{cor:onebit}
Let the database contain $k$ records with $k\geq 3$. Any PIR scheme with a total download of exactly $R+1$ bits requires 1 bit to be downloaded from each of $R$ or $R+1$ different servers in response to any query.
\end{corollary}
\begin{proof}
The special case of Theorem~\ref{thm:binarydownload} when $x=1$ shows that no server replies with more than $1$ bit. For the download complexity to be $R+1$, no more than $R+1$ servers can respond non-trivially. Since the user deduces the value of an $R$-bit record from the bits it has downloaded, at least $R$ servers must reply to any query.
\end{proof}

One might hope that the Corollary~\ref{cor:onebit} could be strengthened to the statement that exactly $R+1$ servers must respond non-trivially. However, examples show that this is not always the case: see the comments after Construction~\ref{con:manyserver} below.

Shah et al.\  state~\cite[Theorem~1]{SRR14} that, in the situation above, ``for almost every PIR operation'' $R+1$ servers must respond, and they provide a heuristic argument to support this statement. The following result makes this rigorous, with a precise definition of `almost every'.

\begin{theorem}
\label{thm:almostall}
Let the database contain $k$ records with $k\geq 3$. Suppose we have a PIR scheme with a total download of exactly $R+1$ bits (in the worst case). Suppose a user chooses to retrieve a record chosen with a uniform probability distribution on $\{1,2,\ldots,k\}$. Let $\alpha$ be the probability that only $R$ bits are downloaded. Then
\[
\alpha\leq \frac{R+1}{kR+1}.
\]
\end{theorem}
\begin{proof}
By Corollary~\ref{cor:onebit}, each server replies to any query with at most one bit. We may assume, without loss of generality, that if a server replies with one bit then this bit must depend on the database in some way (since otherwise we may modify the scheme so that this server does not reply and the probability $\alpha$ will increase).

Let $(q_1,q_2,\ldots,q_n)$ be a query for the $\ell$th record where only $R$ servers reply non-trivially. Since only $R$ servers reply, there are at most $2^R$ possible replies to the query (over all databases). But the value $X_\ell$ of the record is determined by the reply, and there are $2^R$ possible values of $X_\ell$. So in fact there must be exactly $2^R$ possible replies, and there is a bijection between possible replies and possible values $X_\ell$. We claim that the replies of each of these $R$ servers can only depend on $X_\ell$, not on the rest of the database. To see this, suppose a server $S_r$ replies non-trivially, and let $f:\{0,1\}^{kR}\rightarrow\{0,1\}$ be the function mapping each possible value of the database to the reply of $S_r$ to query $q_r$. Suppose $f$ is not a function of $X_\ell$ alone, so there are two databases $\mathbf{X}$ and $\mathbf{X}'$ whose $\ell$th records are equal and such that $f(\mathbf{X})\not=f(\mathbf{X}')$. Let $\rho$ be the common value of the $\ell$th record in both $\mathbf{X}$ and $\mathbf{X}'$. When $X_\ell=\rho$ there are at least two possible replies to the query, depending on the value of the remainder of the database. But this contradicts the fact that we have a bijection between possible replies and possible values $X_\ell$. So our claim follows.

Let $A$ be the event that exactly $R$ servers reply, and for $r=1,2,\ldots,n$ let $B_r$ be the event that server $S_r$ replies non-trivially. Let $D_r$ be the indicator random variable for the event $B_r$. So $D_r$ is equal to $1$ when $S_r$ responds non-trivially and $0$ otherwise. Note that $D_r$ is always equal to the number of bits downloaded from $S_r$, thus the expected value of the sum of these variables satisfies
\begin{equation}
\label{eqn:deqn}
\Exp\left(\sum_{r=1}^n D_r\right)=\alpha R + (1-\alpha)(R+1)=R+1-\alpha.
\end{equation}
Let $D'_r$ be the indicator random variable for the event $A\wedge B_r$. When $A$ does not occur, all the variables $D'_r$ are equal to $0$. When $A$ occurs, $D'_r$ is the number of bits downloaded from server $S_r$ and a total of $R$ bits are downloaded. So
\begin{equation}
\label{eqnddasheqn}
\Exp\left(\sum_{r=1}^n D'_r\right)=(1-\alpha)0+\alpha R=\alpha R.
\end{equation}

Suppose a server $S_r$ uses the following strategy to guess the value of $\ell$ from the query $q_r$ it receives. If the server replies non-trivially using a function $f$ that depends on only one record, say Record $\ell'$, it guesses that $\ell=\ell'$. Otherwise, the server guesses a value uniformly at random.  The server guesses correctly with probability $1/k$ when it responds trivially. The argument in the paragraph above shows the server always guesses correctly if it responds non-trivially and only $R$ servers reply. Thus the server is correct with probability at least $(1/k)\Pr(\overline{B_r})+\Pr(A\wedge B_r)$. The privacy requirement of the PIR scheme implies that the server's probability of success can be at most $1/k$, and so we must have that $\Pr(A\wedge B_r)\leq (1/k)\Pr(B_r)$. Hence
\[
\Exp(D'_r)\leq (1/k)\Exp(D_r).
\]
By linearity of expectation, we see that
\[
\Exp\left(\sum_{r=1}^n D'_r\right)=\sum_{r=1}^n\Exp(D'_r)\leq \frac{1}{k}\sum_{r=1}^n\Exp(D_r)=\frac{1}{k}\,\Exp\left(\sum_{r=1}^n D_r\right).
\]
So, using~\eqref{eqn:deqn} and~\eqref{eqnddasheqn}, we see that
\[
\alpha R\leq \frac{1}{k}(R+1-\alpha).
\]
Rearranging this inequality in terms of $\alpha$, we see that the theorem follows.
\end{proof}

\section{Constructions}
\label{sec:constructions}

Recall the notation from the introduction: we are assuming that our database $\mathbf{X}$ consists of $k$ records, each of $R$ bits, and we write $X_{ij}$ for the $j$th bit of the $i$th record.

\subsection{Two schemes with download complexity $R+1$}
\label{subsec:rplusone}

This section describes two schemes with download complexity $R+1$. Recall that this download complexity is optimal, by Corollary~\ref{cor:rplusone}. The first scheme is included because of its simplicity; it can be thought of as a variation of the scheme of Chor et al.\  described in Example~\ref{ex:chor}, and achieves optimal download complexity using only $R+1$ servers. It has a total storage requirement which is quadratic in $R$. But the scheme has high upload complexity: $kR(R+1)$. The second scheme is very closely related to a scheme mentioned in an aside in Shah et al.~\cite[Section IV]{SRR14}. This scheme has the same properties as the first scheme, except the upload complexity is improved to just $(R+1)k\lceil\log (R+1)\rceil$.

We note that the main scheme described in Shah et al.~\cite[Section IV]{SRR14} also has optimal download complexity of $R+1$. Each server stores just $R$ bits, and so the storage per server is low. However, their scheme uses an exponential (in $R$) number of servers, and so has exponential total storage.

\begin{construction}\label{con:manyserver}
Suppose there are $R+1$ servers, each storing the whole database.
\begin{itemize}
\item A user who requires Record $\ell$ creates a $k\times R$ array of bits by drawing its entries  $\alpha_{ij}$ uniformly and independently at random.
\item Server $S_{R+1}$ is requested to return the bit $c_{R+1}=\bigoplus_{i=1}^k\bigoplus_{j=1}^R\alpha_{ij}X_{ij}$.
\item For $r=1,2,\dotsc, R$, server $S_r$ is requested to return the bit $c_r=\bigoplus_{i=1}^k\bigoplus_{j=1}^R\beta_{ij}X_{ij}$, where
\[
\beta_{ij}=\begin{cases}
\alpha_{ij}\oplus 1&\text{if }i=\ell\text{ and }j=r,\\
\alpha_{i,j}&\text{otherwise}.
\end{cases}
\]
\item To recover $X_{\ell r}$, namely bit $r$ of record $X_\ell$, the user computes $c_r\oplus c_{R+1}$.
\end{itemize}
\end{construction}

\begin{theorem}
\label{thm:construction1}
Construction~\ref{con:manyserver} is a $(R+1)$-server PIR scheme with download complexity $R+1$. The scheme has upload complexity $kR(R+1)$ and total storage $(R+1)Rk$ bits.
\end{theorem}
\begin{proof}
We note that
\[
\alpha_{ij}\oplus\beta_{ij}=\begin{cases}
1&\text{if }i=\ell\text{ and }j=r,\\
0&\text{otherwise}.
\end{cases}
\]
Hence
\begin{align*}
c_r\oplus c_{R+1}&=\bigoplus_{i=1}^k\bigoplus_{j=1}^R(\alpha_{ij}\oplus \beta_{ij})X_{ij}\\
&=X_{\ell r}.
\end{align*}
So the user recovers the bit $X_{\ell r}$ correctly for any $r$ with $1\leq r\leq R$. This proves correctness.

For privacy, we note that $S_{R+1}$ receives a uniformly distributed vector $q_{R+1}=(\alpha_{ij})\in\{0,1\}^{kR}$ in all circumstances. Since the distribution of $q_{R+1}$ does not depend on~$\ell$, no information about~$\ell$ is received by $S_{R+1}$. Similarly, for any $1\leq r\leq R$, the query $q_r=(\beta_{ij})\in\{0,1\}^{kR}$ is uniformly distributed irrespective of the value of $\ell$, and so no information about $\ell$ is received by $S_r$.

We note that each query $q_r$ is $kR$ bits long (for any $r\in\{1,2,\ldots,R+1\}$) and so the upload complexity of the scheme is $kR(R+1)$. Each server replies with a single bit, and so the download complexity is $R+1$. The database is $kR$ bits long, and so (since each server stores the whole database) the total storage is $(R+1)Rk$ bits.
\end{proof}

We note that there are situations where one of the servers is asked for an all-zero linear combination of bits from the database. In this case, that server need not reply. So the number of bits of downloaded in Construction~\ref{con:manyserver} is sometimes $R$ (though usually $R+1$ bits are downloaded). See the comment following Corollary~\ref{cor:onebit}.

We now describe a second construction with improved upload complexity. The construction can be thought of as a variant of Construction~\ref{con:manyserver} where the rows of the array $\alpha$ are all taken from a restricted set $\{e_0,e_1,\ldots,e_{R}\}$ of size $R+1$. A similar idea is used in the constructions in~\cite{SRR14}.

For $i=1,2,\ldots,R$, let $e_i$ be the $i$\textsuperscript{th} unit vector of length $R$. Let $e_0$ be the all zero vector. For binary vectors $\mathbf{x}$ and $\mathbf{y}$ of length $R$, write $\mathbf{x}\cdot \mathbf{y}$ be their inner product; so $\mathbf{x}\cdot \mathbf{y}=\oplus_{j=1}^Rx_jy_j$.

\begin{construction}\label{con:lowerupload}
Suppose there are $R+1$ servers, each storing the whole database.
\begin{itemize}
\item A user who requires Record $\ell$ chooses $k$ elements $a_1,a_2,\ldots,a_k\in\mathbb{Z}_{R+1}$ uniformly and independently at random. For $r=1,\ldots,R+1$, server $S_r$ is sent the vector $q_r=(b_{1r},b_{2r},\ldots,b_{kr})\in\mathbb{Z}_{R+1}^k$, where
\[
b_{ir}=\begin{cases}
a_i+r\bmod R+1&\text{if }i=\ell,\\
a_i&\text{otherwise}.
\end{cases}
\]
\item Server $S_r$ returns the bit $c_{r}=\bigoplus_{i=1}^k e_{b_{ir}}\cdot X_i$.
\item To recover the $j$\textsuperscript{th} bit of $X_\ell$, the user finds the integers $r$ and $r'$ such that $b_{\ell r}=0$ and $b_{\ell r'}=j$. The user then computes $c_r\oplus c_{r'}$.
\end{itemize}
\end{construction}
\begin{theorem}
\label{thm:construction2}
Construction~\ref{con:lowerupload} is an $(R+1)$-server PIR scheme with download complexity $R+1$. The scheme has upload complexity $k(R+1)\log (R+1)$ and total storage $(R+1)Rk$ bits.
\end{theorem}
\begin{proof}
For correctness, we first note that $r$ and $r'$ exist since $b_{\ell r}\in\{0,1,2\ldots,R\}$ takes on each possible value once as $r\in\{0,1,\ldots,R\}$ varies. Also note that
\[
e_{b_{ir}}\oplus e_{b_{ir'}}=\begin{cases}
e_j&\text{if }i=\ell,\\
e_0&\text{otherwise}.
\end{cases}
\]
So, since $e_0=0$,
\[
c_r\oplus c_{r'}=\bigoplus_{i=1}^k(e_{b_{ir}}\oplus e_{b_{ir'}})\cdot X_{i}=e_j\cdot X_\ell
=X_{\ell j}.
\]
So the user recovers the bit $X_{\ell j}$ correctly for any $j$ with $1\leq j\leq R$.

For privacy, we note that $S_r$ receives a uniformly distributed vector $q_{r}\in(\mathbb{Z}_{R+1})^{k}$ in all circumstances. Since the distribution of $q_{r}$ does not depend on~$\ell$, no information about~$\ell$ is received by $S_r$.

The calculations of the total storage and download complexity are identical to those in the proof of Theorem~\ref{thm:construction1}. For the upload complexity, note that it takes just $\log (R+1)$ bits to specify an element of $\mathbb{Z}_{R+1}$. Since each server receives $k$ elements from $\mathbb{Z}_{R+1}$, and since there are $R+1$ servers, the upload complexity of the scheme is $k(R+1)\log (R+1)$ as claimed.
\end{proof}

\subsection{Optimal download complexity for a small number of servers}
\label{subsec:nserver}

For an integer $n$ such that $(n-1)\mid R$, we now describe an $n$-server PIR scheme with download complexity $\frac{n}{n-1}R$ bits. By Theorem~\ref{thm:tdownload}, this construction provides schemes with an optimal download complexity for $n$ servers, provided the number $k$ of records is sufficiently large. This construction is closely related to Construction~\ref{con:lowerupload} above. Indeed, the construction below is a generalisation of Construction~\ref{con:lowerupload} where we work with strings rather than single bits.

We first define an analogue of the bits $e_b\cdot X_i$ computed by servers in Construction~\ref{con:lowerupload}. We divide an $R$-bit string $X$ into $n-1$ blocks, each of size $R/(n-1)$. For $b\in\{1,2,\ldots,{n-1}\}$ we write $\pi_b(X)$ for the $b$th block (so $\pi_b(X)$ is an $R/(n-1)$-bit string). We write $\pi_0(X)$ for the all-zero string $0^{R/(n-1)}$ of length $R/(n-1)$.

\begin{construction}
\label{con:smallnumservers}
Let $n$ be an integer such that $(n-1)\mid R$. Suppose there are $n$ servers, each storing the entire database.\begin{itemize}
\item A user who requires Record $\ell$ chooses $k$ elements $a_1,a_2,\ldots,a_k\in\mathbb{Z}_{n}$ uniformly and independently at random. For $r=1,\ldots,n$, server $S_r$ is sent the vector $q_r=(b_{1r},b_{2r},\ldots,b_{kr})\in\mathbb{Z}_{n}^k$, where
\[
b_{ir}=\begin{cases}
a_i+r\bmod n&\text{if }i=\ell,\\
a_i&\text{otherwise}.
\end{cases}
\]
\item Server $S_r$ returns the $R/(n-1)$-bit string $c_{r}=\bigoplus_{i=1}^k \pi_{b_{ir}}(X_i)$.
\item To recover the $j$\textsuperscript{th} block of $X_\ell$, the user finds the integers $r$ and $r'$ such that $b_{\ell r}=0$ and $b_{\ell r'}=j$. The user then computes $c_r\oplus c_{r'}$.
\end{itemize}
\end{construction}
\begin{theorem}
\label{thm:construction3}
Construction~\ref{con:smallnumservers} is an $n$-server PIR scheme with download complexity $\frac{n}{n-1}R$. The scheme has upload complexity $nk\log n$ and total storage is $nkR$.
\end{theorem}
\begin{proof}
Exactly as in the proof of Theorem~\ref{thm:construction2}, we first note that $r$ and $r'$ exist since $b_{\ell r}\in\{0,1,2\ldots,n-1\}$ takes on each possible value once as $r\in\{0,1,\ldots,n\}$ varies. Also note that when $i\not=\ell$
\[
\pi_{b_{ir}}(X_i)\oplus \pi_{b_{ir'}}(X_i)=\pi_{a_{i}}(X_i)\oplus \pi_{a_{i}}(X_i)=0^{R/(n-1)},
\]
but when $i=\ell$
\[
\pi_{b_{ir}}(X_i)\oplus \pi_{b_{ir'}}(X_i)=\pi_{0}(X_i)\oplus \pi_j(X_i)=\pi_j(X_i)=
\pi_j(X_\ell).
\]
Hence
\[
c_r\oplus c_{r'}=\bigoplus_{i=1}^k(\pi_{b_{ir}}(X_i)\oplus \pi_{b_{ir'}}(X_{i}))=\pi_j(X_{\ell}).
\]
So the user recovers the $j$th block of $X_{\ell}$ correctly for any $j$ with $1\leq j\leq (n-1)$.

For privacy, we note that $S_r$ receives a uniformly distributed vector $q_{r}\in(\mathbb{Z}_n)^{k}$ in all circumstances. Since the distribution of $q_{r}$ does not depend on~$\ell$, no information about~$\ell$ is received by $S_r$.

The total storage is $nkR$, since each of $n$ servers stores the entire $kR$-bit database. Each query $q_r$ is $k\log n$ bits long, since an element of $\mathbb{Z}_n$ may be specified using $\log n$ bits. Hence the upload complexity is $nk\log n$. Since each server returns an $R/(n-1)$- bit string, the download complexity is $\frac{n}{n-1}R$.
\end{proof}

Shah et al.~\cite[Section~V]{SRR14} provide PIR schemes with linear (in $R$) total storage and with download complexity between $2R$ and $4R$. Their scheme requires a number of servers which is independent of $R$ (but is linear in $k$). The construction above shows that for any fixed positive $\epsilon$ a PIR scheme with linear total storage exists with download complexity of $(1+\epsilon)R$ (as we just fix a value of $n$ such that $n/(n-1)<1+\epsilon$). This is within an arbitrarily close factor of optimality. Moreover, the number of servers in our construction is independent of both $k$ and $R$. However, note that in our scheme each server stores the whole database, whereas the per server storage of the scheme of Shah et al.\  is a fixed multiple of $R$. This issue is addressed in Construction~\ref{con:smallperserver} below.

\subsection{Schemes with small per-server storage}
\label{subsec:perserver}

We make the observation that the last construction may be used to give families of schemes with lower per-server storage; see~\cite[Section~V]{SRR14} for similar techniques. The point here is that we never XOR the first bit (say) from one block with the second bit (say) of any other block, so we can store these bits in separate servers without causing problems.

More precisely, let $s$ be a fixed integer such that $s\mid R$ and let $t$ be a fixed integer such that $(t-1)\mid s$. We divide each record $X_i$ into $R/s$ blocks $\pi_1(X_i), \pi_2(X_i),\ldots,\pi_{R/s}(X_i)$, each $s$ bits long. We then divide each block $\pi_j(X_i)$ into $(t-1)$ sub-blocks $\pi_{j,1}(X_i)$, $\pi_{j,2}(X_i),\ldots ,\pi_{j,t-1}(X_i)$, each $s/(t-1)$ bits long. For any $i\in\{1,2,\ldots,k\}$ and any $j\in\{1,2,\ldots,R/s\}$, we define $\pi_{j,0}(X_i)$ to be the all zero string $0^{s/(t-1)}$ of length $s/(t-1)$.

\begin{construction}
\label{con:smallperserver}
Let $s$ be a fixed integer such that $s\mid R$. Let $t$ be a fixed integer such that $(t-1)\mid s$. Let $n=t(R/s)$.  Suppose there are $n$ servers. Each server will store just $ks$ bits.
\begin{itemize}
\item Index the $t(R/s)$ servers by pairs $(u,r)$, where $1\leq r\leq t$ and where $1\leq u\leq R/s$. Server $S_{(u,r)}$ stores the $u$th sub-block of every block. So $S_{(u,r)}$ stores $\pi_{u,j}(X_i)$ where $1\leq i\leq k$ and $1\leq j\leq t-1$. Note that each server stores $k(t-1)s/(t-1)=ks$ bits.
\item A user who requires Record $\ell$ chooses $k$ elements $a_1,a_2,\ldots ,a_k\in\mathbb{Z}_t$ uniformly and independently at random. The server $S_{(k,r)}$ is sent the query $q_r=(b_{1r},b_{2r},\ldots,b_{kr})\in\mathbb{Z}_{t}^k$, where
\[
b_{ir}=\begin{cases}
a_i+r\bmod n&\text{if }i=\ell,\\
a_i&\text{otherwise}.
\end{cases}
\]
(Note that many servers receive the same query.)
\item Server $S_{(u,r)}$ returns the $s/(t-1)$-bit string $c_{(u,r)}=\oplus_{i=1}^k \pi_{u,b_{ir}}(X_i)$.
\item To recover the $j$th sub-block of the $u$th block of $X_\ell$, the user finds integers $r$ and $r'$ such that $b_{\ell r}=0$ and $b_{\ell r'}=j$ and computes $c_{(u,r)}\oplus c_{(u,r')}$.
\end{itemize}
\end{construction}
\begin{theorem}
\label{thm:construction4}
Construction~\ref{con:smallperserver} is a PIR scheme with download complexity $\frac{R}{s}\frac{r}{t-1}s=\frac{t}{t-1}R$. The scheme has upload complexity $nk\log t=(tkR/s)\log t$ and total storage $nks=tkR$ bits.
\end{theorem}
\begin{proof}
As in the proofs of Theorems~\ref{thm:construction2} and~\ref{thm:construction3}, privacy follows since $S_{u,r}$ always receives a uniformly distributed vector $q_r\in\mathbb{Z}_t^k$ as a query. For correctness, observe that when $i\not=\ell$
\[
\pi_{u,b_{ir}}(X_i)\oplus \pi_{u,b_{ir'}}(X_i)=\pi_{u,a_{i}}(X_i)\oplus \pi_{u,a_{i}}(X_i)=0^{s/(t-1)},
\]
but when $i=\ell$
\[
\pi_{u,b_{ir}}(X_i)\oplus \pi_{u,b_{ir'}}(X_i)=\pi_{u,0}(X_i)\oplus \pi_{u,j}(X_i)=\pi_{u,j}(X_i)=
\pi_{u,j}(X_\ell).
\]
Hence
\[
c_{(u,r)}\oplus c_{(u,r')}=\bigoplus_{i=1}^k(\pi_{u,b_{ir}}(X_i)\oplus \pi_{u,b_{ir'}}(X_{i}))=\pi_{u,j}(X_{\ell}).
\]
So the user can indeed compute the $j$-th sub-block of the $u$-th block as claimed.

It is easy to calculate the upload complexity, download complexity and total storage complexity as before, remembering that each server stores $ks$ bits rather than the entire database.\end{proof}

By fixing $t$ and $s$ to be sufficiently large integers, we can see that for all positive $\epsilon$ we have a family of schemes with download complexity at most $(1+\epsilon)R$, with total storage linear in the database size, with a linear (in $R$) number of servers, and where the per server storage is independent of $R$. So this family of schemes has a better download complexity and per-server storage than Shah et al.~\cite[Section~V]{SRR14}, and is comparable in terms of both the number of servers and total storage.

The servers may be divided into $t$ classes $\mathcal{S}_1,\mathcal{S}_2,\ldots,\mathcal{S}_{t}$, where
\[
\mathcal{S}_r=\{S_{(1,r)},S_{(2,r)},\ldots,S_{(R/s,r)}\}.
\]
Since servers in the same class receive the same query, the above construction still works if some of the servers within a class are merged. If this is done, the storage requirements of each merged server is increased, the download complexity and total storage are unaffected, and the number of servers required and upload complexity are reduced. So various trade-offs are possible using this technique.

\subsection{An explicit asymptotically optimal scheme}
\label{subsec:sunjafar}

Sun and Jafar~\cite{SuJa16d} describe a PIR scheme that has the best possible asymptotic download complexity, as $R\rightarrow\infty$. Their scheme is constructed in a recursive fashion. In this subsection, we describe an explicit, non-recursive, scheme with the same parameters as the Sun and Jafar scheme. Our scheme has the advantages of a more compact description, and (we believe) a proof that is significantly more transparent.

Our scheme is described in detail in Construction~\ref{con:sunjafar} below. But, to aid understanding, we first provide an overview of the scheme.

Suppose that $n^k$ divides $R$. We split an $R$-bit string $X$ into $n^k$ blocks, each of length~$R/n^k$. For $j\in\{1,2,\ldots,n^k\}$ we write $\pi_j(X)$ for the $j$-th block of $X$, and we write $\pi_0(X)$ for the all zero block $0^{R/n^k}$.

Let $\mathcal{V}$ be the set of all non-zero strings $\bv=v_1v_2v_3\ldots v_k\in\{0,1,2,\ldots,n-1\}^k$ such that $\sum_{i=1}^k v_i\equiv 0\bmod n-1$. (Note that our sum is taken modulo $n-1$, not modulo $n$.) Let $\mathcal{W}=\{1,2,\ldots,n\}\times \mathcal{V}$. For each record, say Record~$\ell$, we will define a graph $\Gamma^{[\ell]}$ on the vertex set $\mathcal{W}$ (see below).

There are $n$ servers in the scheme, each storing the whole database. Server $S_r$ receives a query consisting of integers $b_i(r,\bv)\in\{1,2,\ldots,n^k\}$ where $i\in\{1,2,\ldots,k\}$ and $\bv\in \mathcal{V}$. The server replies with $|\mathcal{V}|$ strings, each of length $R/n^k$. Each string is a linear combination of blocks, at most one block from each record (the choice of each block being determined by an integer $b_i(r,\bv)$: see~\eqref{eqn:sj_combination} below). From the perspective of $S_r$, the distribution of the integers $b_i(r,\bv)$ does not depend on $\ell$, enabling us to attain privacy. However, the user chooses these integers so that $b_i(r,\bv)$ and $b_{i}(r',\bv')$ are constrained to be equal when $(r,\bv)$ and $(r',\bv')$ lie in the same component of the graph $\Gamma^{[\ell]}$.
This is done in such a way that the user can reconstruct Record $\ell$ from the servers' replies.

\begin{figure}
\begin{tikzpicture}[scale=0.7]
\draw (0,0) node[left]{$(1,101)$}-- (10,0.5); \draw (0,1) node[left]{$(3,202)$} -- (10,0.5) node[right]{$(2,002)$};
\draw (0,2) node[left]{$(1,110)$} -- (10,2.5); \draw (0,3) node[left]{$(3,220)$} -- (10,2.5) node[right]{$(2,020)$};
\draw (0,4) node[left]{$(1,112)$}  -- (10,4.5); \draw (0,5) node[left]{$(3,222)$}  -- (10,4.5) node[right]{$(2,022)$};
\draw (0,6) node[left]{$(1,211)$} -- (10,6.5); \draw (0,7) node[left]{$(2,121)$} -- (10,6.5) node[right]{$(3,011)$};
\draw (0,8)  node[left]{$(2,211)$} -- (10,8.5); \draw (0,9)  node[left]{$(3,121)$} -- (10,8.5)  node[right]{$(1,011)$};

\draw (-1,5) ellipse (3 and 8);
\draw (11,5) ellipse (3 and 6.5);

\filldraw (0,0) circle (0.08);
\filldraw (0,2) circle (0.08);
\filldraw (0,4) circle (0.08);
\filldraw (0,6) circle (0.08);
\filldraw (0,8) circle (0.08);

\filldraw (0,1) circle (0.08);
\filldraw (0,3) circle (0.08);
\filldraw (0,5) circle (0.08);
\filldraw (0,7) circle (0.08);
\filldraw (0,9) circle (0.08);

\filldraw (10,0.5) circle (0.08);
\filldraw (10,2.5) circle (0.08);
\filldraw (10,4.5) circle (0.08);
\filldraw (10,6.5) circle (0.08);
\filldraw (10,8.5) circle (0.08);

\filldraw (0,10) node[left]{$(1,200)$} circle (0.08);

\draw (-1,-0.5) node{$\vdots$};
\draw (11,0) node{$\vdots$};

\draw (-1,14) node{$\mathcal{W}_1^{[\ell]}$};
\draw (11,12.5) node{$\mathcal{W}_2^{[\ell]}$};

\end{tikzpicture}
\caption{Part of the graph $\Gamma^{[\ell]}$ when $n=k=3$ and $\ell=1$.}
\label{fig:sjgraph}
\end{figure}
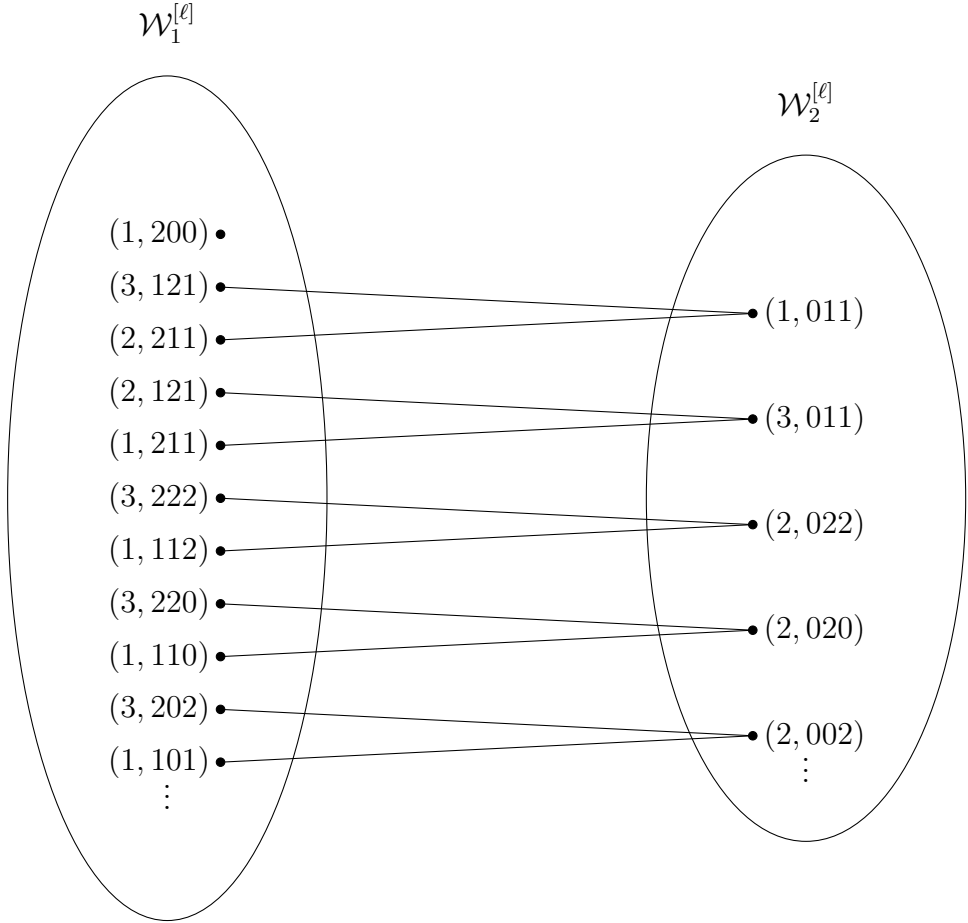

We now give details of the scheme. We begin by describing the graph $\Gamma^{[\ell]}$ (see Figure~\ref{fig:sjgraph}) and by detailing some of its structure.
Let $\ell\in\{1,2,\ldots,k\}$. The graph $\Gamma^{[\ell]}$ is defined on the vertex set $\mathcal{W}$, and is bipartite with parts $\mathcal{W}_1^{[\ell]}$ and $\mathcal{W}_2^{[\ell]}$: the set $\mathcal{W}_1^{[\ell]}$ consists of those elements $(r,\bv)\in \mathcal{W}$ such that $v_\ell\not=0$, and $\mathcal{W}_2^{[\ell]}$ consists of those elements such that $v_\ell=0$. We draw at most one edge from each element $(r,\bv)\in \mathcal{W}_1^{[\ell]}$ into $\mathcal{W}_2^{[\ell]}$ as follows. If $v_\ell$ is the only non-zero entry in $\bv$, we draw no edge from $(r,v_\ell)$, so we have an isolated vertex. Suppose two or more entries of $\bv$ are non-zero.  We define $\ell_2\in\{1,2,\ldots,k\}$ to be the next entry in $\bv$ after the $\ell$th that is non-zero, taken cyclically. Let $w\in\{1,2,\ldots,n-1\}$ be such that $w\equiv v_\ell+v_{\ell_2}\bmod n-1$. Define $\bv'=v'_1v'_2\cdots v'_k$ by
\[
v'_i=\begin{cases}
v_i&\text{if }i\in\{1,2,\ldots,k\}\setminus\{\ell,\ell_2\},\\
0&\text{if }i=\ell,\\
w&\text{if }i=\ell_2.
\end{cases}
\]
Let $r'\in\{1,2,\ldots,n\}$ be such that $r'\equiv r+v_\ell\bmod n$. We join $(r,\bv)$ to $(r',\bv')$.

Let $\mathcal{C}^{[\ell]}$ be the set of connected components of the graph $\Gamma^{[\ell]}$. We note that $\Gamma^{[\ell]}$ has exactly $n$ isolated vertices, namely the vectors of the form $(r,\bv)$ where $r\in\{1,2,\ldots,n\}$ and where $\bv$ is the single vector defined by
\[
v_i=\begin{cases}
0&\text{if }i\not=\ell,\\
n-1&\text{if }i=\ell.
\end{cases}
\]
The remaining components in $\mathcal{C}^{[\ell]}$ are stars consisting of a central vertex in $\mathcal{W}^{[\ell]}_2$ and $n-1$ other vertices all lying in $\mathcal{W}^{[\ell]}_1$. Moreover, we note that if $(r,\bv)$ and $(r',\bv')$ are distinct vertices in the same component of $\Gamma^{[\ell]}$ then $r\not=r'$.

We claim that the number of vertices $(r,\bv)\in\mathcal{W}^{[\ell]}_1$ is $n^k$. To see this, we note that there are $n$ choices for $r$, and then $n^{k-1}$ choices for $v_1,v_2,\ldots,v_{\ell-1},v_{\ell+1},\ldots,v_k$. Once these choices are made $v_\ell\in\{0,1,\ldots,n-1\}$ is determined, since $v_\ell\not=0$ and $\sum_{i=1}^kv_i\equiv 0\bmod n-1$. This establishes our claim.

Since every component of $\Gamma^{[\ell]}$ contains a vertex in $\mathcal{W}^{[\ell]}_1$, we see that $|\mathcal{C}^{[\ell]}|\leq |\mathcal{W}^{[\ell]}_1|=n^k$. Indeed, the number of components of $\Gamma^{[\ell]}$ is:
\[
|\mathcal{C}^{[\ell]}|=n+(|\mathcal{W}^{[\ell]}_1|-n)/(n-1)=n(1+(n^{k-1}-1)/(n-1)).
\]

\begin{construction}
\label{con:sunjafar}
Suppose that $n^k\mid R$. Suppose there are $n$ servers, each storing the whole database.
\begin{itemize}
\item A user who requires Record $\ell$ proceeds as as follows. In the notation defined above, for each $i\in\{1,2,\ldots,k\}\setminus\{\ell\}$ the user chooses (uniformly and independently) a random injection $f_i:\mathcal{C}^{[\ell]} \rightarrow \{1,2,\ldots,n^{k}\}$.  The user chooses (again uniformly and independently) a random bijection $\psi: \mathcal{W}^{[\ell]}_1\rightarrow \{1,2,\ldots,n^{k}\}$.

Define integers $b_i(r,\bv)\in\{0,1,\ldots ,n^k\}$ for $(r,\bv)\in\mathcal{W}$ and $i\in\{1,2,\ldots,k\}$ as follows. If $i\not=\ell$, define
\[
b_i(r,\bv)=\begin{cases}
0&\text{if }v_i=0,\\
f_i(C)&\text{if }v_i\not=0 \text{ and $(r,\bv)$ lies in the component $C\in\mathcal{C}$}.
\end{cases}
\]
Note that when $i=\ell$ we have that $v_i\not=0$ if and only if $(r,\bv)\in\mathcal{W}_1^{[\ell]}$. So when $i=\ell$ we may define
\[
b_i(r,\bv)=\begin{cases}
0&\text{if }v_i=0,\\
\psi((r,\bv))&\text{if }v_i\not=0.
\end{cases}
\]

For $r=1,2,\ldots,n$, server $S_r$ is sent the vector $q_r=(b_i(r,\bv):\bv\in\mathcal{V}, i\in\{1,2,\ldots,k\})$.
\item The server $S_r$ replies with the blocks
\begin{equation}
\label{eqn:sj_combination}
s_{(r,\bv)}=\sum_{i=1}^{k}\pi_{b_i(r,\bv)}(X_i)
\end{equation}
for all $\bv\in\mathcal{V}$.
\item To recover block $j$ of $X_\ell$, the user finds $(r,\bv)=\psi^{-1}(j)\in\mathcal{W}_1^{[\ell]}$. Let $C\in\mathcal{C}^{[\ell]}$ be the component containing $(r,\bv)$. If $|C|>1$, let $(r',\bv')\in C\cap \mathcal{W}_2^{[\ell]}$. Then (see below for justification)
\[
\pi_j(X_\ell)=\begin{cases}
s_{(r,\bv)}&\text{ if }|C|=1,\text{ and}\\
s_{(r,\bv)}\oplus s_{(r',\bv')}&\text{ if }|C|>1.
\end{cases}
\]
\end{itemize}
\end{construction}

\begin{theorem}
\label{thm:construction5}
Construction~\ref{con:sunjafar} is an $n$-server PIR scheme with download complexity $(1-1/n^k)(n/(n-1))R$. The total storage of the scheme is $nkR$. The upload complexity of the scheme is $k^2n^{k}\log n$ bits.
\end{theorem}
\begin{proof}
We begin by establishing correctness of the scheme. Let $(r,\bv)=\psi^{-1}(j)$ and let $C\in \mathcal{C}^{[\ell]}$ be the component containing $(r,\bv)$. When $|C|=1$ we have $v_i\not=0$ if and only if $i=\ell$ and so
\[
s_{(r,\bv)}=\sum_{i=1}^{k}\pi_{b_i(r,\bv)}(X_i)=\pi_{b_\ell(r,\bv)}(X_\ell)=\pi_j(X_\ell),
\]
the last equality following since $b_\ell(r,\bv)=j$. Hence the user recovers the $j$th block $\pi_j(X_\ell)$ of $X_\ell$ correctly in this case. Suppose now that $C$ contains two or more vertices, so there exists $(r',\bv')\in C\cap \mathcal{W}_2^{[\ell]}$. When $i\not=\ell$, the values of $b_i(r,\bv)$ and $b_i(r',\bv')$ are equal, since $(r,\bv)$ and $(r',\bv')$ lie in the same component $C$ of $\Gamma^{[\ell]}$ and since $v_i=0$ if and only if $v'_i=0$. Moreover, $v_\ell\not=0$ and $v'_\ell=0$. Hence
\begin{align*}
s_{(r,\bv)}\oplus s_{(r',\bv')}&=\sum_{i=1}^k \left(\pi_{b_i(r,\bv)}(X_i)\oplus\pi_{b_i(r',\bv')}(X_i)\right)\\
&=\pi_{b_\ell(r,\bv)}(X_\ell)\oplus \pi_{b_\ell(r',\bv')}(X_\ell))\\
&=\pi_{\psi((r,\bv))}(X_\ell)\oplus\pi_{0}(X_\ell))\\
&=\pi_{j}(X_\ell).
\end{align*}
So the user recovers the $j$th block $\pi_j(X_\ell)$ of $X_\ell$ correctly in this case also. We have established correctness.

We now aim to establish the security of the scheme. Let $\mathcal{A}$ be the set of integer vectors $(a_i(\bv)\in\{0,1,\ldots,n^k\}:i\in\{1,2,\ldots,k\}, \bv\in\mathcal{V})$ with the restrictions that $a_i(\bv)=0$ if and only if $v_i=0$, and that for any fixed $i\in\{1,2,\ldots,k\}$ the integers $a_i(\bv)$ with $v_i\not=0$ are distinct.
Let $r\in \{1,2,\ldots,n\}$ be fixed. The query $q_r=(b_i(r,\bv):\bv\in \mathcal{V},i\in\{1,2,\ldots,k\})$ lies in $\mathcal{A}$, since the functions $f_i$ and $\psi$ are injective and since (whether or not $i=\ell$) we have $b_i(r,\bv)=0$ if and only if $v_i=0$. Indeed, the query is uniformly distributed in $\mathcal{A}$. To see this, first note that the functions $f_i$ (for $i\not=\ell$) and $\psi$ are chosen independently. The values $b_\ell(r,\bv)$ for $v_\ell\not=0$ are uniform subject to being distinct since $\psi$ is a randomly chosen bijection. For $i\not=\ell$, the values $b_i(r,\bv)$ for $v_\ell\not=0$ are uniform subject to being distinct, since $f_i$ is a uniformly chosen injection from $\mathcal{C}^{[\ell]}$, and since at most one vertex in any component $C\in\mathcal{C}^{[\ell]}$ has its first entry equal to $r$. Hence the distribution of query $q_r$ is uniform on $\mathcal{A}$ as claimed. Since this distribution does not depend on $\ell$, privacy follows.

Each server replies with $|\mathcal{V}|$ strings, each string of length $R/n^k$. Since there are $n$ servers, the download complexity is $nR|\mathcal{V}|/n^k$. So it remains to determine $|\mathcal{V}|$. For $0\leq s\leq k-1$, there are $n^{k-s-1}$ elements $v_1v_2\cdots v_k\in\mathcal{V}$ that begin with exactly $s$ zeros, since we may choose $v_{s+2},v_{s+3},\ldots,v_{k}\in\{0,1,\ldots,n-1\}$ arbitrarily and then $v_{s+1}$ is determined by the fact it is non-zero and $\sum_{j=1}^kv_j\equiv 0 \bmod n-1$. So
\[
|\mathcal{V}|=\sum_{s=0}^{k-1}n^{k-s-1}=(n^k-1)/(n-1)
\]
and the download complexity is $(1-1/n^k)(n/(n-1))R$, as required.

We may argue that the total upload complexity is $k^2n^{k}\log n$ as follows. Consider Server $S_r$. The integers $b_i(r,\bv)$ with $v_i=0$ are zero, and so do not need to be sent. There are exactly $kn^{k-1}$ integers $b_i(r,\bv)\in\{1,2,\ldots,n^k\}$ with $i\in\{1,2,\ldots,k\}$ and $\bv\in\mathcal{V}$ with $v_i\not=0$. (To see this, note that there are $k$ choices for $i$, and $n$ choices for each component $\bv$ except the $\ell$th. But then $v_\ell$ is determined by the fact that it is non-zero and $\sum_{j=1}^kv_j\equiv 0\bmod n-1$.) Each integer can be specified using $k\log n$ bits, and so the query $q_r$ is $k^2n^{k-1}\log n$ bits long. Since there are $n$ servers, the total upload complexity is $k^2n^{k}\log n$ bits, as required.
\end{proof}

\subsection{An averaging technique}
\label{subsec:averaging}

The download complexity of both the PIR scheme due to Sun and Jafar~\cite{SuJa16d} and the scheme in Construction~\ref{con:sunjafar} above is $(1-1/n^k)(n/(n-1))R$. This is only slightly smaller than the more practical scheme in Construction~\ref{con:smallnumservers}, which has download complexity $(n/(n-1))R$. In fact, the \emph{expected} number of bits downloaded in Construction~\ref{con:smallnumservers} is $(1-1/n^k)(n/(n-1))R$, since a server is asked for an all-zero linear combination of blocks with probability $1/n^k$ and need not reply in this case. This section describes an `averaging' technique which transforms Construction~\ref{con:smallnumservers} into a scheme with good (worst case) download complexity, at the price of a much stronger divisibility constraint on the length of blocks. This technique will work for a wide range of PIR schemes, but in the case of Construction~\ref{con:smallnumservers} it produces a scheme with optimal download complexity $(1-1/n^k)(n/(n-1))R$. Moreover, the upload complexity is considerably smaller than the schemes described in~\cite{SuJa16d} and Construction~\ref{con:sunjafar}.

Before giving the detail, we describe the general idea. Chan, Ho and Yamamoto~\cite[Remark 2]{CHY15} observed that a PIR scheme with good upload complexity (but long record lengths) can be constructed by dividing each record into blocks, then using copies of a fixed PIR scheme for shorter records operating on each block in parallel. Crucially, the same randomness (and so the same queries) can be used for each parallel copy of the scheme, and so upload complexity is low. The `averaging' construction operates in a similar way. However, rather than using the same randomness we use different but predictably varying randomness for each parallel copy. The server can calculate queries for each copy of the scheme from just one query, so upload complexity remains low. But (because queries vary over all possibilities) the resulting scheme has (worst case) download complexity equal to the average number of bits of download in the Chan, Ho and Yamamoto construction.

In more detail, we modify Construction~\ref{con:smallnumservers} as follows. Suppose that $n^k(n-1)\mid R$. We divide an $R$-bit string $X$ into $n^k(n-1)$ blocks, each of size $R/(n^k(n-1))$. We index these blocks by pairs $(b,\bx)$ where $b\in\{1,2,\ldots,n-1\}\subseteq\bZ_n$ and $\bx\in\bZ_n^k$. We write $\pi_{(b,\bx)}(X)$ for the block of $X$ that is indexed by $(b,\bx)$. For any $\bx\in \bZ_n^k$, we write $\pi_{(0,\bx)}(X)$ for the all-zero string $0^{R/(n^k(n-1))}$ of length $R/(n^k(n-1))$.

\begin{construction}
\label{con:averaging}
Let $n$ be an integer such that $n^k(n-1)\mid R$. Suppose there are $n$ servers, each storing the entire database.
\begin{itemize}
\item A user who requires Record $\ell$ chooses $k$ elements $a_1,a_2,\ldots,a_k\in\mathbb{Z}_{n}$ uniformly and independently at random. For $r=1,\ldots,n$, server $S_r$ is sent the vector $q_r=(b_{1r},b_{2r},\ldots,b_{kr})\in\mathbb{Z}_{n}^k$, where
\[
b_{ir}=\begin{cases}
a_i+r\bmod n&\text{if }i=\ell,\\
a_i&\text{otherwise}.
\end{cases}
\]
\item For $r\in\{1,2,\ldots,n\}$ and $\bx\in\bZ_n^k$, define the string $c_{(r,\bx)}$ of length $R/(n^k(n-1))$ by
\[
c_{(r,\bx)}=\bigoplus_{i=1}^k \pi_{(b_{ir}+x_i,\bx)}(X_i).
\]
The server $S_r$ returns the string $c_{(r,\bx)}$, for all $\bx=(x_1,x_2,\ldots,x_k)\in\bZ_n^k$ such that $\bx+q_r\not=\mathbf{0}$. So $S_r$ returns $n^k-1$ strings.
\item To recover the block of $X_\ell$ indexed by a pair $(j,\bx)$, the user finds the integers $r$ and $r'$ such that $b_{\ell r}+x_\ell=0$ and $b_{\ell r'}+x_{\ell}=j$. The user then computes $c_{(r,\bx)}\oplus c_{(r',\bx)}$.
\end{itemize}
\end{construction}
\begin{theorem}
\label{thm:construction6}
Construction~\ref{con:averaging} is an $n$-server PIR scheme with download complexity $(1-1/n^k)\frac{n}{n-1}R$. The scheme has upload complexity $nk\log n$ and total storage is $nkR$.
\end{theorem}
\begin{proof}
We begin with the correctness of the scheme. Exactly as in the proof of Theorem~\ref{thm:construction3}, we note that $r$ and $r'$ exist since $b_{\ell r}+x_\ell\in\{0,1,2\ldots,n-1\}$ takes on each possible value once as $r\in\{0,1,\ldots,n\}$ varies. Moreover, we note that the string $c_{(r,\bx)}$ is all zero if $\bx+q_r=0$ (and similarly the string  $c_{(r',\bx)}$ is all zero if  $\bx+q_{r'}=0$) and so the user always receives enough information to calculate $c_{(r,\bx)}\oplus c_{(r',\bx)}$.

Let $\bx=(x_1,x_2,\ldots,x_k)$. When $i\not=\ell$
\[
\pi_{(b_{ir}+x_i,\bx)}(X_i)\oplus \pi_{(b_{ir'}+x_i,\bx)}(X_i)=\pi_{(a_{i}+x_i,\bx)}(X_i)\oplus \pi_{(a_{i}+x_i,\bx)}(X_i)=0^{R/(n-1)}.
\]
When $i=\ell$
\[
\pi_{(b_{ir}+x_i,\bx)}(X_i)\oplus \pi_{(b_{ir'}+x_i,\bx)}(X_i)=\pi_{(0,\bx)}(X_i)\oplus \pi_{(j,\bx)}(X_i)=\pi_{(j,\bx)}(X_i)=
\pi_{(j,\bx)}(X_\ell).
\]
Hence
\[
c_{(r,\bx)}\oplus c_{(r',\bx)}=\bigoplus_{i=1}^k(\pi_{(b_{ir}+x_i,\bx)}(X_i)\oplus \pi_{(b_{ir'}+x_i,\bx)}(X_{i}))=\pi_{(j,\bx)}(X_{\ell}).
\]
So the user recovers the block of $X_{\ell}$ indexed by $(j,\bx)$ correctly.

Privacy follows from the privacy of Construction~\ref{con:smallnumservers}, as the method for generating queries is identical.

The total storage is $nkR$, since each of $n$ servers stores the entire $kR$-bit database. Each query $q_r$ is $k\lceil \log n\rceil $ bits long, since an element of $\mathbb{Z}_n$ may be specified using $ \log n$ bits. Hence the upload complexity is $nk\log n$. Since there are $n$ servers, and each server returns $n^k-1$ strings of length $R/(n^k(n-1))$, the download complexity is $(1-1/n^k)\frac{n}{n-1}R$.
\end{proof}

\section{Conclusions and future work}
\label{sec:conclusion}

In this paper, we have used classical PIR techniques to prove bounds on the download complexity of PIR schemes in modern models, and we have presented various constructions for PIR schemes which are either simpler or perform better than previously known schemes. The characteristics of the six constructions in this paper are summarised in Fig~\ref{fig:summary}, and parameters for the schemes in~\cite{SRR14} and~\cite{SuJa16b} are included for comparison.

\begin{figure}
{\tiny
\begin{tabular}{c|c|c|c|c|c}
&Download&Upload&Restrictions&Comments\\\hline
\cite{SRR14}&$R+1$&$R(R+1)$&$n=(R-1)^k$&Algorithm~1 and~2 in~\cite{SRR14}\\
\cite{SRR14}&$R+1$&$k(R+1)\log(R+1)$&$n=R+1$&End of~\cite[Sec.~IV]{SRR14}\\
\cite{SRR14}&$(2\Delta/(\Delta-(k-1)))R$&$(\Delta^2/(\Delta-(k-1)))R$&$n\geq 2\Delta$, $\Delta\geq 2k$&Algorithm~3 in~\cite{SRR14}; linear storage\\
\cite{SuJa16b}&$(1-1/n^k)(n/(n-1)R$&$k^2n^k\log n$&$n^k|R$&Optimal asymptotic download; recursive\\
1&$R+1$&$kR(R+1)$&$n=R+1$&Generalisation of~\cite{CGKS98}\\
2&$R+1$&$k(R+1)\log(R+1)$&$n=R+1$&Similar to~\cite[Sect.~IV]{SRR14}; improved expected download\\
3&$\frac{n}{n-1}R$&$nk\log n$&$(n-1)|R$&Optimal download for $n$ servers\\
4&$\frac{t}{t-1}R$&$nk\log t$&$s|R$, $(t-1)|s$, $n=tR/s$&Each server stores only $ks$ bits\\
5&$(1-1/n^k)(n/(n-1)R$&$k^2n^k\log n$&$n^k|R$&Optimal asymptotic download; non-recursive\\
6&$(1-1/n^k)(n/(n-1)R$&$nk\log n$&$n^k(n-1)|R$&Optimal asymptotic download; improved upload\\
\end{tabular}}
\caption{Summary of the six constructions in this paper and those in~\cite{SRR14,SuJa16b}}
\label{fig:summary}
\end{figure}

Various interesting problems remain in this area. We first consider schemes with optimal download complexity:

\begin{problem}
Are there PIR schemes with fewer than $R+1$ bits of download complexity?
\end{problem}
Our paper, like the rest of the literature, only considers PIR schemes over binary channels, and in this model the answer is `no'. But the proofs of this fact in this paper and in Shah at el.~\cite{SRR14} both use the fact that we are working over binary channels: more than $R$ bits of download implies that at least $R+1$ bits are downloaded. So this problem is still open if we extend the model to schemes that do not necessarily use binary channels.

We now return to the standard binary channel model.

\begin{problem}
\label{prob:Shah}
Are there PIR schemes with download complexity $R+1$ and total storage linear in $R$?
\end{problem}
This result was claimed in Shah at el.~\cite{SRR14}, but we believe that a proof of this is still not known. A proof of this result might depend on a more detailed structural analysis of PIR schemes with $R+1$ bits of download. As a first step, we believe the following to be of interest:

\begin{problem}
Theorem~\ref{thm:almostall} bounds the probability that only $R$ bits are downloaded in a PIR scheme with (worst case) download complexity $R+1$. Is this bound tight?
\end{problem}
We conjecture that the bound could be significantly improved in some cases.

We now consider families of schemes that have good asymptotic complexity as $R\rightarrow\infty$.

\begin{problem}
Does there exist a family of schemes with download complexity $(1+o(R))R$ and linear total storage?
\end{problem}
Note that an affirmative solution to Question~\ref{prob:Shah} will imply an affirmative solution to this question.

\begin{problem}
Are there practical PIR schemes that approach asymptotic capacity as $R$ grows?
\end{problem}
The schemes by Sun and Jafar~\cite{SuJa16b} and the related schemes presented in this paper have the strong restriction that $n^k$ must divide $R$.

\begin{problem}
Is there a combinatorial proof that provides a tight upper bound on the asymptotic capacity as $R\rightarrow\infty$?
\end{problem}
We comment that the proof in Sun and Jafar~\cite{SuJa16b} uses information theoretic techniques.
A combinatorial proof might give extra structural information for schemes meeting the bound, and might improve the bound in non-asymptotic cases.

Finally, we turn to larger questions. It is clearly very important to construct schemes with practical parameter sizes, which can work in real-life distributed storage settings. In particular, the following problems are key.

\begin{problem}
Can we find better constructions for PIR schemes?
\end{problem}
Schemes are of interest if they improve per server storage, total storage, upload or download complexity, if the number of servers needed was reduced, or if the divisibility conditions for parameters such as $R$ are weakened.
\begin{problem}
Can the techniques from this paper be applied to establish bounds or give constructions in other models, such as those discussed in Subsection~\ref{subsec:context}?
\end{problem}
In particular, can these constructions be adapted to work when the database is coded (in order to provide robustness against server failure, for example)?

\paragraph{Acknowledgement} The authors would like to thank Doug Stinson for comments on an earlier draft.

\end{document}